\documentclass[11pt,preprintnumbers,pra,aps,nofootinbib,
tightenlines,floatfix,superscriptaddress,longbibliography,
twocolumn]{revtex4-2}
\usepackage[margin=1in]{geometry}
\usepackage[dvipsnames]{xcolor}
\usepackage{amsmath, amssymb, amsthm}
\usepackage{mathrsfs}      
\usepackage{enumerate}     
\usepackage[colorlinks=true,breaklinks=true]{hyperref}
\hypersetup{
linkcolor = TealBlue!90!Black,
citecolor = YellowGreen!90!Black,
urlcolor = Purple!90!Black,
colorlinks = true,
}
\usepackage{tikz}          
\usepackage{graphicx}      
\usepackage{mathtools} 
\usepackage{booktabs}
\usepackage{babel}
\usepackage[normalem]{ulem}
\usepackage{multirow}
\usepackage[section]{placeins}
\theoremstyle{plain}
\newtheorem{theorem}{Theorem}[section]   
\newtheorem{lemma}[theorem]{Lemma}

\newtheorem{corollary}[theorem]{Corollary}

\theoremstyle{definition}

\theoremstyle{remark}

\begin{document}
\title{Susceptibility of qudit-qudit entanglement to quantum noise: insights from the negativity rank}

\author{Matthew Chen}
\email{matthew.chen@duke.edu}
\affiliation{{Duke Quantum Center and Department of Physics, Duke University, Durham, NC 27708, USA}}
\author{Natalie Klco}
\email{natalie.klco@duke.edu}
\affiliation{{Duke Quantum Center and Department of Physics, Duke University, Durham, NC 27708, USA}}
\begin{abstract}
     One aspect of complexity for quantum many-body systems resides in their ability to exhibit entanglement, uniquely quantum correlations central to quantum computation.
To quantitatively explore the diminishing effect that quantum noise has on entanglement resources, we focus on the negativity witness, an entanglement measure based on local time reversal transformations that provides a perfect witness in simple physical contexts.
Beyond calculation of the negativity, we analyze the associated eigenvalue spectra of partially transposed density matrices for qudit-qudit bipartitions, finding that higher negativity rank provides greater(reduced) robustness to targeted(depolarizing) noise.
In particular, we extend the maximum number of separable states that can be arbitrarily mixed into an entangled pure ensemble while retaining negativity to be exactly equal to the original negativity rank.
Furthermore, we report how understanding the role of the negativity rank with respect to the completely depolarizing channel allows direct calculation of (1) certain noise boundaries for the Absolutely Positive Partial Transpose (APPT) subspace beyond which no unitary operator is capable of generating negativity and (2) the maximum qutrit-qutrit APPT purity.
Beyond proving tight bounds and contributing to the ongoing search for APPT entanglement, these results expand our understanding of the nuanced relationship between computation with entanglement and practical environments for operating quantum devices.
\end{abstract}

\date{\today}
\maketitle

{
\small
\tableofcontents
}
\section{Introduction}
Entanglement is a crucial resource in quantum information processing~\cite{PhysRevLett.70.1895, Huelga:1997mw, doi:10.1126/science.1097576, Seif:2024pgs, Bravyi:2025ukf} and provides insight into fundamental correlations within quantum many-body systems and quantum fields~\cite{Reeh:1961ujh, Summers:1987ze, Eisert:2006zz, Hastings:2007iok,  2008PhRvL.101a0504L, Casini:2009sr,   Laflorencie:2015eck, Witten:2018zxz, Klco:2021cxq, Joshi:2023rvd}. Unfortunately, when quantum devices inevitably interact with their environment, the resulting quantum noise can decrease their available entanglement and information complexity~\cite{Zyczkowski:2001hut, Dodd:2003jwv, Yu:2004cga}. Thus, understanding how structural properties of entanglement impact its relationship with various forms of quantum noise is valuable for endeavors to realize the envisioned capabilities of quantum computation. 

For pure states, the entanglement entropy (and a variety of other entanglement measures) provides a necessary and sufficient condition for determining whether a state is entangled or separable across a bipartition~\cite{Horodecki:2009zz}. This is possible because any entropy observed in a subspace of a pure state must arise due to entanglement between the subspace and its complement.
However, for mixed ensembles, subspace entropy could also arise from the classical distribution inherent in the global ensemble itself.
This is one perspective from which to appreciate a fundamental challenge of entanglement detection---any conclusive protocol that has been developed for identifying entanglement between bipartitions of qudit-qudit mixed ensembles is generally computationally infeasible, even in small (e.g., qutrit-qutrit) systems~\cite{HORODECKI19961, Doherty:2003pfm}.
Consequently, many theoretical approaches for imperfectly witnessing entanglement within mixed ensembles have been developed~\cite{ Vidal:1998ch, Sanpera:2000evh, Guhne:2008qic, 10.1145/1993636.1993683, SHI201551, Harrow:2017uji}, providing tradeoffs between computational time and sensitivity.
These approaches are often based on positive but not completely positive (PnCP) maps---linear maps acting on density matrices that preserve positive semidefiniteness when acting globally but may introduce negative eigenvalues when acting locally if entanglement is present. For example, one necessary and sufficient criterion states that a density matrix, $\rho$, is separable if and only if $(\mathbb{I} \otimes \Phi)\rho\geq0$ for all PnCP maps, $\Phi$, where $\mathbb{I}$ is the identity map \cite{HORODECKI19961}. The general infeasibility of this criterion then resides in the open problem of characterizing all PnCP maps in arbitrary dimensions~\cite{Miller:2015hqi}. 

In the following, we focus on the noise robustness information conveyed by a single PnCP map---the transpose $\Phi = T$.  
While this map alone is not necessary and sufficient for witnessing entanglement in generic quantum systems, it does exhibit this significance for small Hilbert spaces (qubit-qubit and qubit-qutrit systems~\cite{HORODECKI19961}) and for simple physical contexts such as spacelike vacuum regions of free scalar fields~\cite{Klco:2021biu, Gao:2024nbl}.  
Furthermore, the absolutely separable (ASEP) subspace---in which no unitary operator is capable of generating entanglement---is equivalent to the absolutely positive partial transpose (APPT) subspace for qubit-qudit~\cite{Johnston:2013wey}  and possibly larger~\cite{Arunachalam:2015ywm} systems,
i.e., APPT is necessary and sufficient for ASEP classification in system sizes where the positive partial transpose (PPT) is no longer sufficient for separability.
Along with these properties of heightened significance, the associated negativity~\cite{Peres:1996dw} of the transpose map---the sum of negative eigenvalues in the partially transposed density matrix $\left( \mathbb{I} \otimes T\right) \rho \equiv \rho^\Gamma$---provides an entanglement measure that has been well studied for its probability distribution across pure states and mixed ensembles~\cite{Datta:2010mda, Shapourian:2020mkc}, its relation to distillibility~\cite{Dur:1999dek, Vidal:2002zz}, as well as its convexity~\cite{Plenio:2005cwa}. Here, we consider a further property, the \emph{negativity rank} $n_-$ defined as the number of negative eigenvalues in $\rho^\Gamma$.  
Beyond $n_->0$ heralding negative partial transpose (NPT) entanglement, this spanning dimension of negativity contributions in the PT space is found to relate to the retention of inseparability in a variety of noisy environments.

This manuscript explores the operational information of response to quantum noise that is conveyed by the negativity rank of an ensemble.
In some cases this means analytic derivations of negativity robustness (\S\ref{sec:negativityDecay},~\ref{sec: IV}), developing intuitive noise-based arguments bounding transitions in the unitary availability of negativity (\S\ref{sec:APPTboundary}), or appreciating the insights and noise implications of known negativity rank bounds (\S\ref{sec:NPTsubspaces}).
In other cases, this means contributing an analytic purity bound calculation  and exploring the impact of recent developments in generating PnCP maps sensitive to PPT entanglement~\cite{CHOI197595, Breuer:2006azi,Hall_2006, Miller:2015hqi, Buckley:2021isn, Masse:2026vps} to deepen appreciation for outstanding problems in mixed-state entanglement classification (\S\ref{sec:maxPurityAPPT33}).

Section~\ref{sec:negativityDecay} begins by establishing the role of negativity rank in governing the rate of negativity decay when subject to the global depolarizing noise channel.
Employing this understanding, Section~\ref{sec:APPTboundary} develops an intuitive derivation for the Bloch radius of the APPT boundary for ensembles radially connected to the Bloch surface, and Section~\ref{sec:maxPurityAPPT33} analytically derives the maximum purity of a qutrit-qutrit APPT state.  Beyond a fundamental result, the latter is shown to provide a computationally amenable softening of constraints for a semidefinite program that may complement the ongoing search for APPT entangled states.
As the relation of Section~\ref{sec:negativityDecay} heightens the practical role of the negativity rank, Section~\ref{sec:NPTsubspaces} recounts an analytic proof of the maximum negativity rank and an explicit maximal NPT subspace $\mathcal{S}$ as a function of the bipartite qudit-qudit Hilbert space dimensions~\cite{PhysRevA.87.064302}.  
In the process of appreciating this derivation, Section~\ref{sec:NPTsubspaces} provides a modicum of additional guidance on efficiently identifying support of negative eigenvalues in $\rho^\Gamma$ for $\rho \in \mathcal{S}$.
Because the maximal negativity rank of mixed ensembles is greater than that for pure states of the same Hilbert space dimensions,  Section~\ref{sec:NPTsubspaces} closes with a discussion of the ability of noise channels to increase the negativity rank, though they do not increase the negativity.
Complementary to Section~\ref{sec:negativityDecay}, Section~\ref{sec: IV} derives 
tight bounds on the number of arbitrary separable (\S\ref{sec:sepmix}) or entangled (\S\ref{sec:entmix}) states that can be mixed into an originally pure NPT ensemble while preserving the NPT entanglement classification.
In particular, the proof for separable mixtures extends the previously known lower bound and connects the exact bound directly to the negativity rank.

\section{Robustness to Depolarization}
\label{sec: II}
\subsection{Negativity Decay Rate}\label{sec:negativityDecay}
The global depolarizing noise channel introduces, with probability $0 \leq p \leq 1$, an unpolarized source into an ensemble, such that a two-qudit density matrix $\rho$ is updated as,
\begin{align}
\label{eq: depolarizingnoise}
    \mathcal{E}_p(\rho)=(1-p)\rho+\frac{p\mathbb{I}}{d^2} \ \ \ ,
\end{align}
where $d^2$ is the dimension of the Hilbert space and $\rho$ is an element in the set of $d^2 \times d^2$ trace-1 Hermitian positive semidefinite matrices. Beyond its theoretical simplicity, the depolarizing channel is relevant in practical applications that leverage twirling techniques~\cite{Zyczkowski:2005jvo, Dankert:2009yux} to globally depolarize otherwise biased sources of noise.
In terms of the Bloch vector $\mathbf{r}$ of $(d^4-1)$ scalar coefficients governing the generalized Pauli $\left( \boldsymbol{\sigma}\right)$ decomposition of a density matrix $\rho(\mathbf{r}) = \frac{\mathbb{I} + \mathbf{r}\cdot{\boldsymbol\sigma}}{d^2}$, the depolarizing channel simply shrinks this vectorization of the ensemble, $\mathcal{E}_p\left( \rho\left( \mathbf{r} \right) \right) = \rho \left( \mathbf{r}(1-p)\right)$.

For quantifying entanglement, this simple rescaling of $\mathbf{r}$ does not extend to the negativity $\mathcal{N}\left(\mathbf{r}\right)$, due to its non-trivial functional relationship with the direction of $\mathbf{r}$.
Rather, the evolution of negativity upon global depolarization is more naturally expressed recursively.
Because the partial transpose (PT) is a linear map, $\langle ij|\rho|k\ell\rangle = \langle i\ell| \rho^\Gamma | kj\rangle$, and the maximally mixed density matrix is invariant under partial transposition, $\left(\mathbb{I}/d^2\right)^\Gamma = \mathbb{I}/d^2$, it is equivalent to partially transpose before or after the global depolarizing noise channel,
\begin{equation}
\mathcal{E}_p\left( \rho\right)^\Gamma = \mathcal{E}_p\left( \rho^\Gamma\right)    \ \ \ .
\end{equation}
Since the negativity is the sum of negative eigenvalues in the spectrum of $\rho^\Gamma$, 
\begin{equation}
    \mathcal{N}_i=(1-p)\mathcal{N}_{i-1}-\frac{pn_-}{d^2}  \ \ \ ,
    \label{eq:recurrence}
\end{equation}
valid when $p$ is sufficiently small that $n_-\left( \rho^\Gamma \right) = n_-\left( \mathcal{E}_p\left(\rho^\Gamma\right) \right)$, i.e., when the negativity rank---the number of negative PT eigenvalues---is unchanged.  
Thus, the rate of negativity reduction under depolarizing noise is governed by its rank $n_-$, with higher rank associated with more rapid decay.  For the  continuous depolarizing channel, a qudit pair would experience discrete softening of the negativity decay rate as $n_-$ crosses thresholds of integer decrements.

\subsection{Qudit-Qudit APPT Boundary}
\label{sec:APPTboundary}

For pure states, any separable state can be transformed to be maximally entangled by application of an appropriate entangling unitary.
Mixed ensembles, however, exhibit constrained entanglement generation~\cite{Abellanet-Vidal:2026cij}.
In fact, there exists an absolutely separable (ASEP) subspace for which no unitary can produce entanglement~\cite{Zyczkowski:1998yd,Gurvits:2002yud,Johnston:2013wey,Abellanet-Vidal:2024jnc}.
A convenient stepping stone to characterizing the ASEP subspace is the absolutely PPT (APPT~\footnote{Because unitary transformations leave spectra invariant, these subspaces are also referred to as separable from spectrum (SEPFS) and PPT from spectrum (PPTFS).}) subspace, in which no unitary can produce a form of entanglement that violates the PPT condition~\cite{Hildebrand:2005eni,Arunachalam:2015ywm}.
This convenience is due to the presence of a necessary and sufficient condition in terms of matrix inequalities that are computable in small dimensions for determining membership within the APPT subspace~\cite{Hildebrand:2005eni}. 
As the boundaries of ASEP and APPT subspaces mark distinct changes in the quality of entanglement available for information processing, it is valuable to know how much depolarization a quantum computation can withstand before crossing them.

\begin{figure}
\includegraphics[width=0.5\textwidth]{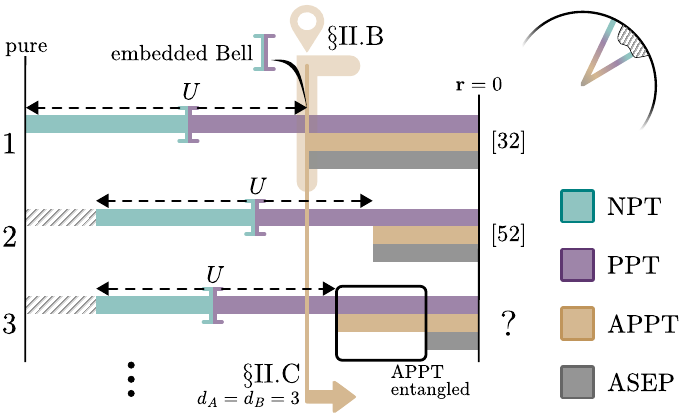}
\caption{Depolarized trajectories of entanglement classifications in the two-qudit ($d^4-1$)-dimensional Bloch sphere of radius $\left|\mathbf{r}_{\text{pure}}\right| = \sqrt{d^2-1}$ for pure states and for mixed ensembles that are radially inaccessible from the surface. The $U$ indicates that unitarily related trajectories exhibit different radii of their NPT/PPT boundaries. Type 1 trajectories depict the technique carried out in Section~\ref{sec:APPTboundary} to find the APPT boundary for the eigenspectrum associated with pure states and its corresponding depolarized trajectories. This technique can be generalized for any trajectory such as the type 2 trajectory described in Ref.~\cite{Song:2024svl}, which is an extremal point of the APPT subspace. In qutrit-qutrit systems, type 1 trajectories also saturate the upperbound on the purity of an ensemble belonging to the APPT subspace, as derived in Section~\ref{sec:maxPurityAPPT33}. Type 3 trajectories, whose existence is unknown, contain APPT entangled ensembles, i.e., their APPT boundary does not correspond to their ASEP boundary.}
\label{fig:apptSketch}
\end{figure}
With the understanding of negativity decay under depolarizing noise established in Section~\ref{sec:negativityDecay}, the APPT boundary can be straightforwardly calculated for mixtures expressible as depolarized pure states, $\mathbf{r} = (1-p) \mathbf{r}_{\text{pure}}$.
To do so, identify a pure state that reaches its PPT boundary at the smallest Bloch vector norm upon depolarization.
From Eq.~\eqref{eq:recurrence}, this state will have the smallest possible negativity rank ($n_- = 1$) and the largest magnitude of its single negative PT eigenvalue ($\lambda^\Gamma_- = \frac{1}{2}$, as proven in Ref.~\cite{PhysRevA.87.054301}). 
As depicted in Fig.~\ref{fig:apptSketch}, the two-level Bell state embedded in any qudit dimensions has these properties.  
By setting $\mathcal{N}_1 = 0$, $\mathcal{N}_0 = 1/2$, and $n_- = 1$ in Eq.~\eqref{eq:recurrence} leading to a single-application channel strength of $p_0=\frac{d^2}{d^2+2}$, the Bloch vector norm at which the Bell state becomes PPT under the effect of depolarizing noise is found to be
\begin{equation}
    \left|\mathbf{r}_0\right|= p_0 \left|\mathbf{r}_{\text{pure}}\right| = \frac{2}{d^2+2}\sqrt{d^2-1} \\\ .
\end{equation}
Since all pure states are unitarily connected to one another and the depolarizing channel commutes with all unitary transformations, the image of pure states under depolarizing noise is also unitarily connected. 
Therefore, the $\left|\mathbf{r}_0\right|$ boundary derived from the embedded Bell state coincides with the APPT boundary for all ensembles connected radially to the Bloch surface. 
This set of ensembles includes the well-studied isotropic states, and $\left|\mathbf{r}_0\right|$ agrees with the standard APPT boundary calculated for them~\cite{Arunachalam:2015ywm}.  

Furthermore, the spectrum of the APPT ensemble resulting from the decay of the Bell state under depolarizing noise turns out to be an extremal point of the APPT set~\cite{Song:2024svl}. Since the APPT set is a convex set, it is generated by the convex hull of extreme points, which are defined as points that cannot be written as a convex combination of other points in the set. Thus, we have provided a simple derivation of a critical APPT bound for an important family of states in arbitrary qudit-qudit dimensions through analysis of the negativity evolution under depolarizing noise. 
\subsection{Maximum Purity of Qutrit-Qutrit APPT Ensemble}
\label{sec:maxPurityAPPT33}

It is interesting that the ASEP and APPT subspaces are equivalent for qubit-qudit
states $\text{ASEP}_{2\times d} = \text{APPT}_{2\times d}$~\cite{Johnston:2013wey} despite the strict separation of the separable and PPT subspaces for qubit-qudit states with $d >3$.  
While progress continues in understanding the relation between these absolute subspaces, their distinction in higher dimensions currently remains unknown~\cite{Arunachalam:2015ywm}.  
If such a distinction arises, an entangled APPT state would necessarily exist.  
In order to develop an additional strategy for excluding an entanglement witness's ability to identify entanglement in an APPT state, focusing on the $3\times 3$ dimensional case,  the following presents an analytic derivation of the maximum purity of a qutrit-qutrit APPT state.

For qutrit-qutrit states, there are two matrix inequalities that determine APPT classification. 
\begin{theorem}
    Let $\vec{\lambda}$ be the ordered eigenvalues $\lambda_1 \geq \lambda_2 \geq \cdots \geq \lambda_9 \geq 0$ of qutrit-qutrit density matrix, $\rho$. Then,
    \begin{subequations}\label{eq:APPTcriteria}
    \begin{align}
L_1\equiv    
\begin{pmatrix}
2\lambda_{9} & \lambda_{8}-\lambda_{1} & \lambda_{6}-\lambda_{2} \\
\lambda_{8}-\lambda_{1} & 2\lambda_{7} & \lambda_{5}-\lambda_{3} \\
\lambda_{6}-\lambda_{2} & \lambda_{5}-\lambda_{3} & 2\lambda_{4}
\end{pmatrix}
\succeq 0,
\\
L_2\equiv
\begin{pmatrix}
2\lambda_{9} & \lambda_{8}-\lambda_{1} & \lambda_{7}-\lambda_{2} \\
\lambda_{8}-\lambda_{1} & 2\lambda_{6} & \lambda_{5}-\lambda_{3} \\
\lambda_{7}-\lambda_{2} & \lambda_{5}-\lambda_{3} & 2\lambda_{4}
\end{pmatrix}
\succeq 0
    \end{align}
    \end{subequations}
are both satisfied if and only if $\rho$ is $APPT$~\cite{Hildebrand:2005eni}. 
\end{theorem}
\noindent
In other words, the absence of negative eigenvalues in $L_1$ and $L_2$ is a necessary and sufficient condition for $\rho$ to be in the APPT subspace.

To begin our calculation of the maximum purity, note that $L_1$ and $L_2$ can have maximally one negative eigenvalue.
\begin{theorem}\label{thm: negeigL}
    If $\Vec{\lambda}$ is an ordered vector of qutrit-qutrit density matrix eigenvalues, then $L_1$ and $L_2$ of Eq.~\eqref{eq:APPTcriteria} have at most one negative eigenvalue. 
\end{theorem}
\noindent
This theorem can be proven by either directly calculating the eigenvalues or analyzing the sign of determinants of $L_1$ and $L_2$, as discussed in Appendix~\ref{app: Proof of Theorem II.2}. 

Next, we determine a bound on the magnitude of the single negative eigenvalue of $L_{1,2}$.
\begin{corollary}
\label{cor: II.3}
    Let $\mu_{\min}$ be the smallest eigenvalue of $L_1$ or $L_2$, then $|\mu_{\min}|\leq1$
\end{corollary}
\begin{proof}
    The Gershgorin circle theorem~\cite{Gershgorin1931} states that the eigenvalues of a matrix, $A$, lie within at least one of the Gershgorin discs, $D(a_{ii}, R_i)$ where $a_{ii}$ is a diagonal element of $A$ defining the center of the $i^{\text{th}}$ disc and $R_i=\sum_{j \neq i} |a_{ij}|$ defines its radius. Applying this theorem to $L_1$ and $L_2$ gives the following lowerbounds on each negative eigenvalue. For $L_1$, the lowest disk yields
        \begin{equation}
        \mu_{\text{min}}\geq -(\lambda_1+\lambda_2)+\lambda_6+\lambda_8+2\lambda_9 \ \ \ .
    \end{equation}
    Similarly for $L_2$,
    \begin{equation}
        \mu_{\text{min}}\geq -(\lambda_1+\lambda_2)+\lambda_7+\lambda_8+2\lambda_9 \ \ \ .
    \end{equation}
    While a stronger $\vec{\lambda}$-dependent bound could be composed, these disks are consistent with $|\mu_{\text{min}}| \leq 1$, due to the $\text{tr}\left[\rho\right] = \sum_i \lambda_i = 1$ normalization.  The saturation of the lowerbound, $\mu_{\text{min}} = -1$, is achieved when $\vec{\lambda}=(1,0,0,0,0,0,0,0,0)$. Therefore, $\left|\mu_{\text{min}}\right| \geq -1$ provides a $\vec{\lambda}$-independent upperbound and tight lowerbound on the $L_{1,2}$ negative eigenvalue. 
\end{proof}
The goal of this section is to relate the purity of a state to  the value of its $L_1$ and $L_2$ negative eigenvalues. First, we must prove some intermediary Lemmas. 

\begin{lemma}\label{lem: II.4}
    Let $\lambda_1$ be the maximal eigenvalue of a qutrit-qutrit density matrix and $P$ be the purity. Then $\lambda_1\geq \frac{1+\sqrt{{8(9P-1)}}}{9}$.
\end{lemma}
\begin{proof}
    The purity is the sum of squares of density matrix eigenvalues, $P= \sum_i \lambda_i^2$. With $\lambda_1$ fixed, $\sum_i \lambda_i^2$ is minimized when the rest of the eigenvalues are all degenerate, i.e. $\lambda_2=\cdots =\lambda_9$. From the normalization condition, $\sum_i \lambda_i = 1$, this means that they must be equal to $\frac{1-\lambda_1}{8}$. Hence, 
    \begin{equation}
        P\geq \lambda_1^2 +\frac{(1-\lambda_1)^2}{8}  \\\ ,
    \end{equation} which may be inverted to provide a lower bound on the largest eigenvalue as a function of the purity,\begin{equation}
        \lambda_1 \geq \frac{1+\sqrt{8(9P-1)}}{9} \\\ ,
    \end{equation}
    by selecting the solution that satisfies the ordered hierarchy of eigenvalues $\lambda_1 \geq \lambda_{2, \dots, 9}$ for $P > P_{\text{min}} = \frac{1}{9}$.
\end{proof}
\noindent
Note that this Lemma generalizes to arbitrary $d$ as $\lambda_1 \geq \frac{1}{d^2}\left(1+\sqrt{(d^2-1)(d^2P-1)}\right)$, though this section will remain focused on the APPT structure of $d = 3$ pairs.

Continuing from Lemma~\ref{lem: II.4}, the minimum eigenvalues of $L_{1,2}$ can be related to the purity. 
\begin{lemma}
\label{lem: II.5}
    Let $P$ be the purity of density matrix $\rho$. Then the minimum eigenvalue of $L_1$ or $L_2$, $\mu_{\text{min}} = \min\left[\mu_{\text{min},L_1},\mu_{\text{min},L_2}\right]$, is upperbounded $\mu_{\text{min}} \leq \frac{8-11\sqrt{2(9P-1)}}{36}$. 
\end{lemma}
\begin{proof}
    According to the Rayleigh quotient, $\lambda_{\text{min},M}  \leq \frac{v^T M v}{v^Tv}$ for a Hermitian matrix $M$ and a nonzero vector $v$ with equality being attained when $v$ is the eigenvector associated with the minimum eigenvalue, $\lambda_{\text{min},M}$. Therefore, $\mu_{\min,L_1} \leq \frac{v^T L_1 v}{v^T v}$ and $\mu_{\min,L_2} \leq \frac{v^T L_2 v}{v^T v}$ for arbitrary 3-vector $v$. Let $v=\frac{1}{\sqrt{2}}(1,1,0)$. Then the Rayleigh quotient for $L_1$ is $\lambda_9+\lambda_8+\lambda_7-\lambda_1$ and for $L_2$ is $\lambda_9+\lambda_8+\lambda_6-\lambda_1$. Since the eigenvalues are ordered and the maximum values of the smallest eigenvalues correspond to the case of degenerate $\lambda_{2,\dots,9} = \frac{1-\lambda_1}{8}$, 
    \begin{equation}\lambda_9+\lambda_8+\lambda_7-\lambda_1\leq \frac{3(1-\lambda_1)}{8}-\lambda_1
    \end{equation}
    and \begin{equation}
        \lambda_9+\lambda_8+\lambda_6-\lambda_1\leq \frac{3(1-\lambda_1)}{8}-\lambda_1 \\\ .
    \end{equation}
    Therefore, an upper bound of the minimum eigenvalue for both $L_1$ and $L_2$ is
    \begin{equation}
        \mu_{\min} \leq \frac{3-11\lambda_1}{8} \\\ ,
    \end{equation}
    which is saturated for the flat spectrum $\lambda_{1, \dots, 9}=1/9$.
    Combining this bound with the result of Lemma~\ref{lem: II.4}, 
    \begin{equation}
        \mu_{\min} \leq \frac{8-11\sqrt{2(9P-1)}}{36} \\\ .
        \label{eq: mumin}
    \end{equation}
While this bound is not tight in general, it is tight for the case $\lambda_{2,\dots,9} = \frac{1-\lambda_1}{8}$.
\end{proof}

We are now prepared to prove the main theorem of this section, which upperbounds the purity of a qutrit-qutrit APPT ensemble.
\begin{theorem}
\label{thm: theoremII.6}
For a qutrit-qutrit ensemble $\rho \in \text{APPT}$, the purity is upperbounded as $\text{tr}\left[ \rho^2\right] \leq P_{\text{max}}^{\text{APPT}} =  \frac{17}{121}$. 
\end{theorem}
\begin{proof}
    The purity at which the upperbound of Eq.~\eqref{eq: mumin} vanishes, crossing from positive to negative, is $P_0 = \frac{17}{121}$. Thus, APPT states do not exist with $P>P_0$. Because there exists a state that saturates the Eq.~\eqref{eq: mumin} upperbound, $P_0=P_{\text{max}}^{\text{APPT}}$ provides the maximum of the APPT purity. An APPT spectrum that realizes this maximum is $\lambda_{2 ,\ldots,9} = \frac{1-\lambda_1}{8}$ with $\lambda_1 = \frac{3}{11}$.
\end{proof}
\begin{proof}
    An alternative proof in the style of Section~\ref{sec:APPTboundary} can be crafted, recognizing the saturating spectrum above as that of depolarized pure states.
    Consider the effect of depolarizing noise on $L_1$ and $L_2$. The global depolarizing noise channel acting on a qutrit-qutrit ensemble $\rho$ is defined as in Eq.~\eqref{eq: depolarizingnoise} with $d = 3$. Every ensemble reaches its APPT boundary at some value of $p<1$ under the action of a global depolarizing channel. Therefore, if  the final purity at which a state reaches the APPT boundary under the depolarizing channel is written as a function of the starting purity, then the local maximum of this function on the allowed domain will be the maximum purity of an APPT ensemble. 
    
    Notice that the eigenvalues of $\rho$ under the depolarizing noise are updated as 
    \begin{equation}
        \mathcal{E}_p(\lambda)=(1-p)\lambda+\frac{p}{9} \ \ \ ,
    \end{equation}
    with $L_{1,2}$ updating accordingly,

    \begin{equation}
        \mathcal{E}_p(L_{1,2})=(1-p)L_{1,2}+\frac{2p}{9} \mathbb{I} \ \ \ .
    \end{equation}
    Therefore, the eigenvalues of $L_1$ under the action of depolarizing noise are updated as 
    \begin{equation}
        \label{eq: depolmu}
        \mathcal{E}_p(\mu)=(1-p)\mu+\frac{2p}{9} \ \ \ .
    \end{equation}
    Using Theorem \ref{thm: negeigL}, the $p$-value depolarizing the ensemble to its APPT boundary can be found at the critical point $\mathcal{E}_{\bar{p}}(\mu_{\text{min}}) \equiv 0$, where $\mu_{\text{min}}$ is the single negative eigenvalue of $L_1$ or $L_2$,
\begin{equation}
\label{eq: pmu}
\bar{p} = \frac{9 \mu_{\text{min}}}{9\mu_{\text{min}}-2} \ \ \ .
\end{equation}
    The effect of the depolarizing noise on the purity of $\rho$ is given by 
    \begin{equation}
    \label{eq: depolpurity}
        P_f=(1-p)^2 P_0 +\frac{2p-p^2}{9} \\\ ,
    \end{equation}
    where $P_0$ is the initial purity and $P_f$ is the final purity. Evaluation of $P_f\left(p=\bar{p}\right)$ produces an expression for the purity at the corresponding APPT boundary as a function of the initial purity $P_0$ and the initial negative eigenvalue $\mu_\text{min}$ of $L_1$ or $L_2$. Maximization of $P_f\left(p=\bar{p}\right)$ subject to the domains $\frac{1}{9}\leq P_0 \leq 1$, $-1\leq\mu_{\text{min}}\leq 0$ of Corollary~\ref{cor: II.3}, and the correlated bound of Lemma~\ref{lem: II.5} results also in a maximum APPT purity of $P_{\text{max}}^{\text{APPT}}  = \frac{17}{121}$.  Observe that this maximum is achieved beginning from any pure state, for which $P_0 = 1$ and $\mu_\text{min} = -1$. 
    \end{proof}
    
    The value derived in Theorem~\ref{thm: theoremII.6} is consistent with the numerical evaluation of Ref.~\cite{Dung:2025xjg}. Additionally, it provides a positive resolution to Conjecture 6.7 in Ref.~\cite{Ahiable:2026how} in the qutrit-qutrit case. Since the depolarized spectrum that realizes this maximum is also a member of ASEP (e.g., as determined by Refs.~\cite{Arunachalam:2015ywm, Vidal:1998ch}), Theorem~\ref{thm: theoremII.6} is also the maximum purity of an ASEP ensemble in qutrit-qutrit systems. 

The alternate proof of Theorem~\ref{thm: theoremII.6} also provides a strategy to directly calculate the APPT boundary of any qutrit-qutrit spectrum. Given an eigenvalue spectrum corresponding to a family of unitarily related density matrices, the effect of depolarizing noise on the associated negative eigenvalues of $L_{1,2}$ can be calculated with Eq.~\eqref{eq: depolmu}. Then, setting the left-hand side of Eq.~\eqref{eq: depolmu} to zero determines the exact boundary at which the eigenvalues become positive, corresponding to the Bloch radius where the depolarized spectrum reaches the APPT boundary. 
    
    Led by Ref.~\cite{Arunachalam:2015ywm}, the most promising strategy thus far developed for informing whether the ASEP and APPT subspaces are equivalent in qutrit-qutrit or higher dimensions has been proving via semidefinite programs (SDPs) that a variety of PnCP maps cannot detect APPT entanglement. 
Several such proofs have relied upon strategically weakening necessary and sufficient SDP conditions to aid tractability (e.g., Lemma 2 of Ref.~\cite{Arunachalam:2015ywm}). 
Employing the APPT purity upperbound of Theorem~\ref{thm: theoremII.6} as an SDP constraint provides a complementary weakened form for determining whether an entanglement witness can detect APPT entanglement. Although entanglement witnesses---$\mathcal{W}_{\chi, \Phi} = (\mathbb{I} \otimes \Phi^\dagger) |\chi\rangle\langle \chi|$ capable of observing entanglement when $\text{tr}\left[\mathcal{W}_{\chi, \Phi} \rho\right] <0$---cannot recover the full detection power of their associated PnCP map $\Phi$, this purity bound can guide numerical searches for APPT entangled ensembles, combined with Proposition 3 of Ref.~\cite{Masse:2026vps} to focus on $|\chi\rangle$ with maximal Schmidt rank or Proposition 1 of Ref.~\cite{Arunachalam:2015ywm} to focus on $\rho$ with full rank. 

For example, an SDP inspecting the ability of an entanglement witness to detect entanglement at purities below $P_{\text{max}}^{\text{APPT}}$ allowed us to find witnesses capable of detecting low-purity PPT entanglement associated with a recent extremal atomic positive map~\cite{Miller:2015hqi}, distinguished from its predecessors by its two-dimensional commutative stable subalgebra. 
These witnesses associated with the map are also found to violate the conditions in Lemma 2 of Ref.~\cite{Arunachalam:2015ywm}, suggesting that $\Phi_{\text{\scriptsize\cite{Miller:2015hqi}}}$ may be able to detect APPT entanglement. 
However, none of the entanglement witnesses sampled could detect APPT entanglement (evaluated rigorously via Equation 3 of Ref.~\cite{Arunachalam:2015ywm}). 
Of course, this numerical search over promising witnesses does not constitute a proof of whether $\Phi_{\text{\scriptsize\cite{Miller:2015hqi}}}$ is capable of detecting APPT entanglement.
Though these observations give further insight into the challenge of identifying a qutrit-qutrit APPT entangled ensemble, determining whether ASEP$_{3\times3}$ and APPT$_{3\times3}$ are equivalent subspaces remains elusive.

\section{Maximal Rank NPT Subspaces}\label{sec:NPTsubspaces}

As discussed in Section \ref{sec:negativityDecay}, the maximum number of negative PT eigenvalues supported by a bipartite Hilbert space, $\max \, n_-$, sets the maximum rate of depolarization-induced negativity decay. 
The upperbound on $\max\,n_-$ for a $d_A \times d_B$ dimensional Hilbert space has been shown to be $(d_A-1)(d_B-1)$. This upperbound was initially proven in Ref.~\cite{PhysRevA.87.054301} and was shown to be tight in Ref.~\cite{PhysRevA.87.064302}.  Following the logic of the latter, this bound on $n_-$ is governed by the maximum dimension of an NPT subspace. 
\begin{theorem}
    There exists a subspace $\mathcal{S} \subset \mathcal{H}_{d_A} \otimes \mathcal{H}_{d_B}$ of dimension $(d_A-1)(d_B-1)$ such that every density matrix with range contained in $\mathcal{S}$ has a negative partial transpose~\cite{PhysRevA.87.064302}.
\end{theorem} 
\begin{proof}
The following simply adds intermediate detail to the complete proof of Theorem 1 in Ref.~\cite{PhysRevA.87.064302}, which develops a particular maximal NPT subspace
\begin{multline}
    \mathcal{S} \equiv \operatorname{span}\left\{ \frac{|j\rangle |k+1\rangle-|j+1\rangle|k\rangle}{\sqrt{2}}: \right.\\ \left. 0\leq j \leq d_A-2, 0 \leq k \leq d_B-2 \right\} \ \ .
    \label{spanning}
\end{multline}
While these basis vectors are not orthogonal, their linear independence spans $(d_A-1)(d_B-1)$ dimensions.
A state within this subspace can thus be characterized by a $(d_A-1)(d_B-1)$-dimensional vector $v$ as
\begin{equation}
|v\rangle = \frac{1}{\sqrt{N_0}} \sum_{jk} v_{jk} |\mathcal{S}_{jk}\rangle \ \ \ ,
\end{equation}
with a normalization factor 
\begin{equation}
N_0 = |\mathbf{v}|^2 - \sum_{j=0}^{d_A-3} \sum_{k=1}^{d_B-2} \Re\left( v_{jk} v*_{(j+1)(k-1)}\right) \ .
\label{eq:normalization}
\end{equation}
This relation frames the unit-norm condition on states in Hilbert space that superpositions of the
spanning vectors of $\mathcal{S}$ must satisfy. 

Rather than confirming the NPT classification of density matrices in $\mathcal{S}$ through eigenvalue calculation, this approach proceeds by proving that the partial transpose of such ensembles always contains a two-dimensional principal submatrix with negative determinant~\cite{PhysRevA.87.064302}, i.e., for all $\rho \equiv\sum_i p_i |v_i\rangle \langle v_i| \in \mathcal{S}$ there exist two-dimensional bipartite coordinates $\mathbf{a},\mathbf{b}$ such that $\det\left( \rho^\Gamma_{\mathbf{a}\mathbf{b}} \right) < 0$.
This is a sufficient condition for the presence of a negative PT eigenvalue.

As Hermiticity persists upon partial transposition, the principle submatrix $\rho^\Gamma_{\mathbf{a}\mathbf{b}}$ has off-diagonal entries related by a complex conjugate. 
In the case that $\mathbf{a},\mathbf{b}$ are chosen to make the upper left $\rho^\Gamma_{\mathbf{a}\mathbf{b}}$ element vanish, the submatrix determinant will be non-positive, $\det\left( \rho_{\mathbf{a}\mathbf{b}}\right) \leq 0$. 
Thus, this approach translates to proving that there exists $\mathbf{a},\mathbf{b}$ that simultaneously satisfies $\langle a_0 b_0| \rho^\Gamma | a_0 b_0\rangle=0$ and $\langle a_0 b_0| \rho^\Gamma | a_1 b_1\rangle\neq0$.
While the first condition would be trivially satisfied by $a_0 = b_0 = 0$ as $|00\rangle \notin \mathcal{S}$, the second condition cannot always be satisfied with this choice, e.g., the state $|v\rangle  = |\mathcal{S}_{01}\rangle$ has a single $2 \times 2$ principal submatrix of negative determinant located at $\mathbf{a} = \left(0,1\right)$ and $\mathbf{b} = \left(1,2\right)$ in its partially transposed density matrix.

To analyze the simultaneous viability of both these conditions, it is convenient to relate this $\rho^\Gamma_{\mathbf{a}\mathbf{b}}$ principal submatrix to a non-principal submatrix of a matrix embedding of the pure states within the convex decomposition of $\rho$.
 
First, characterizing the $\rho^\Gamma_{\mathbf{a}\mathbf{b}}$ submatrix, note that any state $|v\rangle \in \mathcal{S}$ can be written in the computational basis, 
\begin{equation}
|v_i\rangle = \sum_{j,k,\ell} v^{(i)}_\ell c^{(\ell)}_{jk} |j\rangle|k\rangle \ \ \ ,
\end{equation}
where $c^{(\ell)}$ is a $d_Ad_B$-dimensional vector containing two non-zero elements $\pm 1/\sqrt{2}$ corresponding to the $\ell^{\text{th}}$ binary-indexed spanning vector of $\mathcal{S}$, and the normalization factor of Eq.~\eqref{eq:normalization} has been embedded into $v$ from now on.
This leads to a partially transposed density matrix of the form
\begin{multline}
\left(|v_i\rangle\langle v_i|\right)^\Gamma
= \sum_{j,k,\ell,j',k',\ell'} 
v_\ell^{(i)} \overline{v_{\ell'}^{(i)}} 
c_{jk}^{(\ell)} \overline{c_{j'k'}^{(\ell')}}  \\
|j\rangle\langle j'| \otimes |k'\rangle\langle k| \ \ \ ,
\end{multline}
with matrix elements composed of sum products 
\begin{multline}
\langle j k| \left(|v_i\rangle\langle v_i|\right)^\Gamma | j' k'\rangle  = \\  \left(\sum_{\ell} v_\ell^{(i)} c_{jk'}^{(\ell)} \right) \left(\sum_{\ell'}  \overline{v_{\ell'}^{(i)}} \
\overline{c_{j'k}^{(\ell')}}\right) 
\label{eq:rhoGelements}
\end{multline}
indexed by combined $|jk\rangle$ bipartite coordinates, which span $|a_0b_0\rangle$ and $|a_1b_1\rangle$ in the $\mathbf{a}, \mathbf{b}$ subspace.  

Second, consider an isomorphism $\phi: \mathbb{C}^{d_A} \otimes \mathbb{C}^{d_B} \rightarrow M_{d_B \times d_A}$ as the linear map that sends the computational basis state $|i\rangle \otimes |j\rangle$ to the matrix, $|j\rangle\langle i|$. This map is familiar from the proof of the Schmidt decomposition, though more commonly transposed.  
Acting this isomorphism on a state in $\mathcal{S}$ yields
\begin{equation}
\phi(|v_i\rangle) = \sum_{j,k,\ell} v_\ell^{(i)} c_{jk}^{(\ell)} |k\rangle\langle j| \ \ \ .
\end{equation}
A (generally non-principal) submatrix characterized by $\mathbf{a},\mathbf{b}$ may be selected as the two columns of $\mathbf{a}$ and two rows of $\mathbf{b}$,
\begin{align}
\phi(|v_i\rangle)_{\mathbf{a}\mathbf{b}}&=
\begin{pmatrix}
\sum_\ell v_\ell^{(i)} c_{a_0 b_0}^{(\ell)} 
& \sum_\ell v_\ell^{(i)} c_{a_1 b_0}^{(\ell)} \\
\sum_\ell v_\ell^{(i)} c_{a_0 b_1}^{(\ell)} 
& \sum_\ell v_\ell^{(i)} c_{a_1 b_1}^{(\ell)}
\end{pmatrix}
\\
&\equiv
\begin{pmatrix}
0 & \beta_i \\
\alpha_i & \gamma_i
\end{pmatrix} \ \ \ .
\end{align}
Referring to Eq.~\eqref{eq:rhoGelements},
these $\phi(|v_i\rangle)_{\mathbf{a}\mathbf{b}}$ submatrix elements are recognized as components of the $\rho^\Gamma_{\mathbf{a}\mathbf{b}}$ principal submatrix,
\begin{equation}
\rho^\Gamma_{\mathbf{a}\mathbf{b}} = \sum_i p_i \begin{pmatrix}
0  & \alpha_i \overline{\beta_i} \\
\overline{\alpha_i} \beta_i & |\gamma_i|^2
\end{pmatrix} ,
\end{equation}
 with a determinant of
\begin{equation}
\label{eq: negdet}
    -\left|\sum _i p_i \overline{\alpha}_i \beta_i\right|^2 \leq 0 \ \ \ .
\end{equation}
To then prove that such a subspace is NPT, it is sufficient to show that Eq.~\eqref{eq: negdet} becomes a strict inequality, i.e., for every ensemble $\rho \in \mathcal{S}$ there exists $\mathbf{a},\mathbf{b}$ such that $\sum_i p_i\overline{\alpha}_i \beta_i \neq 0$. 

Notice that because the $\mathcal{S}$ spanning states of Eq.~\eqref{spanning} are superpositions of two computational basis states with constant sum of $A$- and $B$-indices, $\phi\left(\mathcal{S}\right)$ is a set of $d_B \times  d_A$ matrices that have two non-vanishing matrix elements $\pm 1/\sqrt{2}$ located within a single antidiagonal, such that each of the $(d_A+d_B-1)$ antidiagonals sum to 0. 
Since this isomorphism is linear, $\phi\left(|v\rangle\right)$ continues to have antidiagonals that sum to zero $\forall |v\rangle \in \mathcal{S}$. 
These properties exhibited by the image of states in $\mathcal{S}$ allowed Ref.~\cite{PhysRevA.87.064302} to construct a succinct proof-by-contradiction to show that there must exist a set of indices, e.g., located along the leftmost antidiagonal that is nonzero for at least one $\phi\left(|v_i\rangle\right)$ with $a_0b_1$ chosen to be its leftmost nonzero element, that satisfy $\sum_i p_i\overline{\alpha}_i \beta_i \neq 0$.
\end{proof}

For pure states $|v\rangle \in \mathcal{S}$, exact locations of $2 \times 2$ negative-determinant $\rho^\Gamma$ submatrices can be quickly identified by inspection of the state's image under the action of $\phi$.  For example, pairs of populated entries of the leftmost nonzero antidiagonal have locations corresponding to $a_0b_1$ and $a_1b_0$, defining $2\times 2$ subspaces of the partially transposed density matrix that contain a negative eigenvalue.

For mixed ensembles $\rho \in \mathcal{S}$, the location of a $2\times 2$ negative-determinant $\rho^\Gamma$ subspace can be declared solely from inspection of the $\phi(|v_i\rangle)$ images only in limited applications, such as rank-2 mixtures in $d = 3$, due to the cancellations that can arise through the classical mixture.
In general, one can modestly extend the strategy developed in Ref.~\cite{PhysRevA.87.064302} to set $a_0 b_1$ (or $a_1 b_0$) by the leftmost (or rightmost) nonzero element of the shortest nonzero antidiagonal over the images of pure states in the ensemble and calculate $\sum_i p_i \bar{\alpha}_i \beta_i$ with $a_1b_0$ (or $a_0b_1$) defined by each other element of this antidiagonal.
This constrained search can be completed in a number of $\sum_i p_i \bar{\alpha}_i \beta_i$ evaluations that is one fewer than the length of the smallest nonzero antidiagonal, at most $d-1$.

Along with this theorem, Ref. \cite{PhysRevA.87.064302} shows by construction via SDP that there exists an ensemble $\rho_{\max n_-} \in \mathcal{S}$ that saturates the negativity rank upperbound of $(d_A-1)(d_B-1)$.
This number is strictly greater than the maximum negativity rank of pure quantum states---recalling that a pure state of Schmidt rank $s\leq\min\left(d_A, d_B\right)$ has $n_- = \frac{s(s-1)}{2}$ negative eigenvalues in its partial transpose.
From Eq.~\eqref{eq:recurrence}, global depolarizing noise will always decrease the negativity rank. 
Therefore, mixed ensembles with more negative eigenvalues than the upperbound for pure states will never be reached by depolarizing. 
To show that this inaccessibility is not a general property of quantum noise, the following constructs a simple example of a channel that can increase the negativity rank beginning from a pure ensemble. 
Let $|\lambda_i\rangle$ be the orthonormalized eigenstates of $\rho_{\max n_-}$. 
Consider the $d_Ad_B$ Kraus operators 
\begin{subequations}
\label{eq:KrausRho9}
\begin{align}
	K_0 &= \sqrt{\lambda_0} \ \mathbb{I} \\
    K_{i>0}&=\sqrt{\lambda_i} \Biggl( |\lambda_i \rangle \langle\lambda_0|+|\lambda_0 \rangle \langle\lambda_i|+ \Biggr. \\ \Biggl. & \hspace{2.5cm} \sum^{d_Ad_B-1}_{j \neq 0, i} |\lambda_j \rangle \langle\lambda_j| \Biggr) \ \ \ , \nonumber
\end{align}
\end{subequations}
leading to a completely positive trace preserving (CPTP) map, $\sum_j K_j^\dagger K_j = \mathbb{I}$, in the form of $\mathcal{E}(\rho) \equiv 	\sum_j K_j\rho K_j^\dagger$.
This channel performs $\mathcal{E}\left( |\lambda_0\rangle\langle \lambda_0|\right) = \rho_{\max n_-}$, increasing the negativity rank regardless of the particular unique or degenerate eigenvector chosen to be $|\lambda_0\rangle$. 
Though the negativity itself is monotonically decreasing under CPTP maps \cite{Eisert:2006kue}, this monotonicity does not extend to the negativity rank.

Unlike that of Eq.~\eqref{eq:KrausRho9}, many common noise channels are local, i.e.,  can be decomposed into channels acting on the A and B Hilbert spaces individually. 
While the negativity rank is found to be nonincreasing under many of the common local noise channels, this limited directionality is also not a general feature of local noise channels. 
For example, by sampling random Kraus operators with the techniques developed by Ref.~\cite{Kukulski:2021uos}, it is common to find local noise channels that are capable of increasing the negativity rank in qubit-qutrit ensembles.
For a concrete qudit-qudit example, consider a pair of Kraus operators, $K_0 = \frac{1}{\sqrt{2}}\mathbb{I}$ and $K_1 = \frac{1}{\sqrt{2}}\left(U_+^2 \otimes U_+^2\right)$, where $U_+$ is the cyclic raising operator in the single-qudit Hilbert space.
Applied with $d\geq 4$ to $\rho$ a Bell state embedded in neighboring qudit levels, $n_-\left(\rho\right) = 1$ increases to $n_-\left( \mathcal{E}\left[ \rho\right]\right) = 2$ as entanglement is probabilistically shifted to a second subspace above the original in the periodically ordered qudits. 
Though the negativity spectrum is fully invariant under local unitary operations, its negativity rank is neither stable nor monotonic under LOCC.

\section{Robustness to Statewise Mixture}
\label{sec: IV}

While higher negativity rank was shown in Section~\ref{sec:negativityDecay} to correspond to more rapid negativity decay under the global depolarizing noise channel, higher negativity rank can also correspond to greater robustness of negativity in the presence of more targeted noise channels.  
In this section, we determine tight bounds on the number of arbitrary separable (\S\ref{sec:sepmix}) or entangled (\S\ref{sec:entmix}) states that can be mixed with an NPT pure state while retaining the NPT classification of the ensemble.
\begin{figure}[t]
\includegraphics[width=0.48\textwidth]{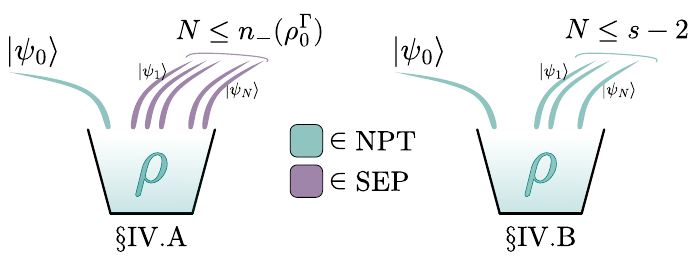}
\caption{Diagram of Section~\ref{sec: IV} results. The maximum number of arbitrary separable (left) or entangled (right) pure ensembles that can be mixed into that of $|\psi_0\rangle$ such that the resulting ensemble $\rho \in \text{NPT}$ is governed by the $|\psi_0\rangle$ Schmidt rank $s$, with the former equal to the negativity rank $n_-\left(\rho_0^\Gamma\right) = \frac{s(s-1)}{2}$. }
\label{fig:mixdiagram}
\end{figure}
Subject to this NPT retention constraint, as summarized in Fig.~\ref{fig:mixdiagram}, 
these two channels can maximally introduce a number of pure ensembles that scales linearly and as a square-root, respectively, of the original state's negativity rank.

\subsection{Separable Mixtures}
\label{sec:sepmix}
Suppose $|\psi_0\rangle \in \text{NPT}$ is a bipartite entangled pure state of two qudits. For an arbitrary set of separable states $\{|\psi_i\rangle \in \text{SEP for} \, i=1,...,N\}$, consider the following mixed state
\begin{equation}
    \rho=\lambda_0|\psi_0\rangle\langle\psi_0|+\sum_{i=1}^N\lambda_i|\psi_i\rangle\langle\psi_i| \ \ \ 
  \label{eq:mixture}
\end{equation}
for arbitrary probabilities satisfying $\sum_{i=0}^{N}\lambda_i = 1$.
The question we are interested in solving is: what is the largest value of $N$ for which $\rho$ remains NPT. 
Previous studies have established a lower bound to this maximal $N$ of $\frac{s(s-1)}{2}-1$, where $s$ is the Schmidt rank of the initial entangled pure state $|\psi_0\rangle$ \cite{Hua:2013sos}. 
In this section, we extend this bound by one and show that in fact $N$ can be as large as $\frac{s(s-1)}{2}$, i.e., $N$ is inclusively upperbounded by the negativity rank of $|\psi_0\rangle$. Note that this bound is now tight, as previous studies have given specific counterexamples to any further extension of the bound \cite{Hua:2013sos}.

To begin, consider derivation of the previous lowerbound to the maximal $N$, Theorem 2 of Ref.~\cite{Hua:2013sos},
emphasizing its restriction to the introduction of separable states. 
\begin{theorem} \label{thm:Theorem 1}
A mixed ensemble in the form of Eq.~\eqref{eq:mixture} composed of an entangled pure state $|\psi_0\rangle \in \text{NPT}$ with Schmidt rank $s$ and $N\leq \frac{s(s-1)}{2}-1 = n_-\left(\rho_0^\Gamma\right)-1$ additional arbitrary pure separable states is NPT. 
\end{theorem}

\begin{proof}
   Since $\rho_0^\Gamma \equiv (|\psi_0\rangle \langle \psi_0|)^\Gamma$ is Hermitian, its eigenvectors are orthogonal. Therefore, the eigenvectors associated with negative eigenvalues will span a partially transposed (PT) subspace 
of dimension $n_-\left(\rho_0^\Gamma\right) = \frac{s(s-1)}{2}$. 
Because PT density matrices of pure separable states remain rank-1 (see Eqs.~\eqref{eq:psiotimes}--\eqref{eq:psiotimesGamma}), the PT space spanned by the $N$ arbitrary separable states will be at most $N$.
In this way, $\rho^\Gamma$ can be appreciated as a weighted sum of two Hermitian matrices: $\rho_0^\Gamma$ and a positive semidefinite matrix of rank upperbounded by $N$. 
Via inertia analysis~\cite{Marquesdesa1981143}, the rank of the PT separable components in the mixture allowed by the present theorem is thus insufficient to eliminate all negative eigenvalues,
$n_-\left(\rho^\Gamma\right) \geq n_-\left(\rho_0^\Gamma\right) - N \geq 1$, such that $\rho \in \text{NPT}$.  

\end{proof}
\noindent Note that a key property of this proof---preservation of rank-1 structure for PT pure separable state density matrices---does not extend to entangled states. 
For example, $\text{rank}\left[\left(|\Psi\rangle\langle \Psi|\right)^\Gamma\right] = 4$, where $|\Psi\rangle$ is a two-level Bell state embedded in arbitrary dimensions.
This increase in rank of entangled state PT density matrices leads to a parametrically smaller upperbound on $N$ when $\left\{|\psi_i\rangle \in \text{NPT}\right\}$, as discussed in Section~\ref{sec:entmix}. 

Observe that the proof of Theorem~\ref{thm:Theorem 1} does not rely on the purity of $\rho_0$, and thus immediately extends to initially entangled mixed ensembles.  
\begin{corollary}\label{cor:rho0mixedApp}
A mixed ensemble $\rho$ composed of an entangled mixed ensemble $\rho_0 \in \text{NPT}$ and $N\leq  n_-\left(\rho_0^\Gamma\right)-1$ additional arbitrary pure separable states is NPT. 
\end{corollary}
\noindent The subsequent question of extending the upperbound on $N$ in Theorem~\ref{thm:Theorem 1} to the negativity rank of $|\psi_0\rangle$ thus relates to the ability of separable ensembles to compensate for the $n_-\left(\rho_0^\Gamma\right)$ negativity contributing eigenvectors of $\rho_0^\Gamma$.
While the former are tensor product in the PT space, the latter are not.
This structural incompatibility will be found to allow extension of the upperbound beyond general inertia or rank analyses.  

To see this incompatibility arising when $\rho_0$ is pure, we prove that the negativity rank $n_-\left( \rho_0^\Gamma\right)$ is invariant upon mixture with a single separable pure state.
\begin{lemma}\label{lem:key}
A mixed ensemble in the form of Eq.~\eqref{eq:mixture} with $N$ = 1 composed of  $|\psi_0\rangle \in \text{NPT}$ and $|\psi_\otimes\rangle  \in \text{SEP}$ will retain its negativity rank $n_-\left( \rho_0^\Gamma\right) = n_-\left(\rho^\Gamma\right)$.
\end{lemma}

\begin{proof}
Consider a pure entangled quantum state written in its Schmidt decomposition form after local transformations as 
\begin{equation}
    |\psi_0\rangle = \sum_{j = 0}^{s-1} \mu_j |jj\rangle \ \ \ ,
\label{eq:psi0Schimdt}
\end{equation}
with $\sum_{j=0}^{s-1}\mu_j^2=1$, $\mathbf{\mu} \in \mathbb{R}^s$, $\mu_j \geq 0 \ \forall j$,  and $s\geq 2$ for the state to be entangled. 
The two Hilbert spaces across the bipartition will be regarded as qudits each with dimension $d \geq s$.
The partially transposed density matrix has a clear set of eigenvalues and eigenvectors captured by the spectral decomposition 
\begin{multline}
\rho_0^\Gamma = \sum_{j = 0}^{s-1} \mu_j^2 |jj\rangle \langle jj|   + \\ \sum_{j = 0}^{s-2} \sum_{k = j+1}^{s-1}\mu_j \mu_k |\psi_{jk}^+\rangle \langle \psi_{jk}^+|  - \\ \sum_{j = 0}^{s-2} \sum_{k = j+1}^{s-1}\mu_j \mu_k |\psi_{jk}^-\rangle \langle \psi_{jk}^-| \ \ \ ,
\label{eq:spectraldecomp}
\end{multline}
where $|\psi_{jk}^\pm\rangle= \frac{|jk\rangle \pm |kj\rangle}{\sqrt{2}}$.
Constructing a unitary that diagonalizes $\rho_0^\Gamma$ in the order of positive, negative, and vanishing eigenvalues with content outside the $s^2$-dimensional subspace results in,
\begin{widetext}
\begin{equation}
    U \rho_0^\Gamma U^\dagger \equiv D =  \text{diag} \left( \overbrace{\mu_0^2, ..., \mu_{s-1}^2}^{\text{$s$}}, \right.  \overbrace{\mu_0\mu_1, ..., \mu_{s-2}\mu_{s-1}}^{\text{$s(s-1)/2$}}, \underbrace{-\mu_0\mu_1, ..., -\mu_{s-2}\mu_{s-1}}_{\text{$s(s-1)/2$}}, \left. \underbrace{0, ..., 0}_{\text{$d^2-s^2$}}\right) \ \ \ .
    \label{eq:ddelements}
\end{equation}
\end{widetext}
Clearly, leading principal submatices of $D$ with dimension $p$ have determinants that are positive for the first $p \leq s+\frac{s(s-1)}{2}$, alternating in sign for the next $\frac{s(s-1)}{2}$ values of $p$, then zero for leading principal submatices of dimension $p> s^2$.

The mixed ensemble of present focus, $\rho = \lambda_0 \rho_0 + \lambda_1 \rho_\otimes $, yields a PT density matrix
\begin{equation}
    \rho^\Gamma = \lambda_0 \rho_0^\Gamma +\lambda_1 \rho_\otimes^\Gamma 
\end{equation}
that adds a positive semidefinite (PSD) matrix to $\rho_0^\Gamma$.  
Beyond its PSD property, $\rho_\otimes^\Gamma$ for an arbitrary separable state in the $d^2$ two-qudit Hilbert space
\begin{equation}
    |\psi_{\otimes}\rangle = \left( \sum_{j = 0}^{d-1} a_j |j\rangle \right) \left( \sum_{k = 0}^{d-1} b_k |k\rangle \right)
\label{eq:psiotimes}
\end{equation}
is a PT density matrix that remains rank-1
\begin{equation}
  \rho_\otimes^\Gamma =   \left(|\psi_\otimes\rangle \langle \psi_\otimes|\right)^\Gamma = |\psi_\otimes^\Gamma\rangle \langle \psi_\otimes^\Gamma| \ \ \ , 
\end{equation}
governed by a single vector that is tensor product in the PT space
\begin{equation}
    |\psi_\otimes^\Gamma\rangle \equiv \left( \sum_{j = 0}^{d-1} a_j |j\rangle \right) \left( \sum_{k = 0}^{d-1} b_k^* |k\rangle \right) \ \ \ .
\label{eq:psiotimesGamma}
\end{equation}

Because unitarily equivalent matrices have equal spectra,  the PT eigenvalues can be analyzed by first transforming into the PT eigenbasis of the entangled state,  
\begin{equation}
    M \equiv U\rho^\Gamma U^\dagger = \lambda_0 D + \lambda_1 v v^\dagger
    \label{eq:UrhoGammaUdg}
\end{equation}
with 
\begin{widetext}
\begin{multline}
    v = U|\psi_\otimes^\Gamma\rangle =  \left( \overbrace{a_0 b_0^*, ...,  a_{s-1} b_{s-1}^*}^{\text{$s$}}, \right. \overbrace{\frac{a_0b_1^* + a_1 b_0^*}{\sqrt{2}}, ..., \frac{a_{s-2}b_{s-1}^* + a_{s-1} b_{s-2}^*}{\sqrt{2}}, }^{s(s-1)/2} \\
    \overbrace{\frac{a_0b_1^* - a_1 b_0^*}{\sqrt{2}}, ..., \frac{a_{s-2}b_{s-1}^* - a_{s-1} b_{s-2}^*}{\sqrt{2}}, }^{s(s-1)/2} 
    \overbrace{a_0b_s^*, a_s b_0^*, ..., a_{s-1} b_d^*, a_d b_{s-1}^*}^{2s(d-2)}, \left.
    \overbrace{a_{s}b_s^*, ..., a_d b_d^*}^{(d-s)^2} \right)^T  \ \ \ .
    \label{eq:vvelements}
\end{multline}
\end{widetext}
By design, the first $s+\frac{s(s-1)}{2}$ principle submatrices of $M$ are positive definite by virtue of the same in $\lambda_0D$ and the positivity of all principle submatrices of $\lambda_1 v v^\dagger$. 
To analyze the sign of the next $\frac{s(s-1)}{2}$ leading principal submatrix determinants, consider the complex matrix determinant lemma 
\begin{equation}
    \det \left( A+u w^\dagger\right) = \det \left( A\right) \left( 1 + w^\dagger A^{-1} u\right) \ \ \ ,
\end{equation}
where $A$ is the leading principal submatrix of $\frac{\lambda_0}{\lambda_1}D$ and $u=w$ is the corresponding leading segment of $v$.
Applying this to the mixture of Eq.~\eqref{eq:UrhoGammaUdg}
indicates that the oscillating signs in $D$ will be reproduced by $M$ if $\left(1+\frac{\lambda_1}{\lambda_0} v_p^\dagger D_p^{-1} v_p\right) \geq 0$  for  $p\leq s^2$, where again $p$ indicates the size of the leading principal submatrix
(Note that $(D_p)^{-1} = \left(D^{-1}\right)_p$ due to the diagonal structure of $D$).

Because $\frac{\lambda_1}{\lambda_0}$ can be arbitrarily large, the asserted positivity must apply to $v_p^\dagger D_p^{-1} v_p$ alone.  
By working in the PT eigenbasis of the entangled state, this inner product may be simply expressed as 
\begin{equation}
v_p^\dagger D_p^{-1} v_p = \sum_{j= 0}^{p} D^{-1}_{jj} |v_j|^2 \ \ \ ,
\label{eq:positiveReq}
\end{equation}
with the matrix and vector elements of Eq.~\eqref{eq:ddelements} and Eq.~\eqref{eq:vvelements}, respectively.
In this form, it is again clear that the first $s+\frac{s(s-1)}{2}$ leading principal submatrices will continue to be positive definite.
Focusing next on the leading principal submatrix of size $p = s^2$, the positivity of
\begin{equation}
    v_{s^2}^\dagger D_{s^2}^{-1} v_{s^2} = \frac{\left| \sum\limits_{j = 0}^{s-1} a_j b_j \prod\limits_{\stackrel{k=0}{k\neq j}}^{s-1} \mu_k\right|^2 }{ \left( \prod\limits_{\ell =0}^{s-1} \mu_\ell^2\right) }
    \label{eq:s2Req}
\end{equation}
shows that the number of negative PT eigenvalues in the $s^2$ subspace of $\rho_0$ and $\rho$ are either both even or both odd, i.e., $n_{-}\left(D_{s^2}\right) \text{mod}\ 2= n_{-}\left(M_{s^2}\right) \ \text{mod} \ 2  $. Due to the rank-1 property of $\rho_\otimes^\Gamma$ allowing the shift of the inertia to be maximally one~\cite{Marquesdesa1981143}, this observation is already sufficient to conclude invariance of the negativity rank, 
\begin{equation}
    n_- \left( \rho^\Gamma\right) = n_-\left(\rho_0^\Gamma\right)  = \frac{s(s-1)}{2} \\\ .
\end{equation}
This invariance is consistent with the observation that all negativity contributing terms of Eq.~\eqref{eq:spectraldecomp} have rank greater than one.

For an additional perspective from Eq.~\eqref{eq:positiveReq}, evaluating $v_p^\dagger D_p^{-1} v_p$ for leading principal submatrices of sizes $s+\frac{s(s-1)}{2}<p<s^2$ will only add strictly positive terms to the evaluation of Eq.~\eqref{eq:s2Req}, e.g., $\frac{\left|a_{s-2}b_{s-1}^* - a_{s-1} b_{s-2}^*\right|^2}{2 \mu_{s-1}\mu_{s-2}}$ for $p = s^2-1$.
This reflects that Eq.~\eqref{eq:s2Req} is positive while including all possible negative contributions.
Thus, the oscillatory signs of sequential principal determinants for dimensions $s+\frac{s(s-1)}{2}\leq p \leq s^2$ persists upon mixture of a single separable state.
According to the interlacing theorem, this  results in the number of negative PT eigenvalues remaining invariant up to dimension $s^2$. 
Because $v v^\dagger \in \text{PSD}$, the number of negative eigenvalues cannot increase via this mixture~\cite{Marquesdesa1981143}, concluding that the negativity rank is invariant for the full $d^2$ Hilbert space of the mixed ensemble including $|\psi_0\rangle$ of Schmidt rank $s$ and an arbitrary separable state.  
\end{proof}

\begin{theorem}\label{thm:maxNsep}
A mixed ensemble in the form of Eq.~\eqref{eq:mixture} composed of an entangled pure state $|\psi_0\rangle \in \text{NPT}$ with Schmidt rank $s$ and $N\leq \frac{s(s-1)}{2} = n_-\left(\rho_0^\Gamma\right)$ additional arbitrary pure separable states is NPT.
\end{theorem}
\begin{proof}
   Lemma~\ref{lem:key} shows that a mixed ensemble of $|\psi_0\rangle\in \text{NPT}$ and one arbitrary pure separable state has a negativity spectrum that retains the entangled state's negativity rank, $\frac{s(s-1)}{2}$.  
By Corollary~\ref{cor:rho0mixedApp},  another $\frac{s(s-1)}{2}-1$ arbitrary pure separable states may be further included in the mixed ensemble while retaining its NPT classification. 
Therefore, an ensemble of $|\psi_0 \rangle$ and $\frac{s(s-1)}{2}$ arbitrary pure separable states will be NPT, extending by one the lowerbound on the maximal $N$ reported in Theorem~\ref{thm:Theorem 1} established in Ref.~\cite{Hua:2013sos}. 
\end{proof}
\subsection{Entangled Mixtures}
\label{sec:entmix}
The previous section considered the NPT classification of a mixture of arbitrary separable states with a single entangled state.
This section will explore the NPT classification when the mixture of Eq.~\eqref{eq:mixture} includes entangled states $|\psi_i\rangle \in \text{NPT}$.

Because the rank of positive components in the PT density matrix of an entangled state exceeds that of a separable state (rank-1), entangled states have greater capacity to reduce negativity through classical mixture. 
A simple example of this is an equal-probability mixture of two two-level Bell states with the same bit-parity, which results in a separable ensemble.
Therefore, when $s=2$, the minimum number of entangled states capable of eliminating the entanglement through classical mixture is one. 
More generally, this minimum number scales linearly with the Schmidt rank.
\begin{theorem} \label{thm:maxNent}
A mixed ensemble in the form of Eq.~\eqref{eq:mixture} composed of entangled pure state $|\psi_0\rangle$ with Schmidt rank $s$ and $N\leq s-2$ additional arbitrary entangled pure states is NPT.
\end{theorem}
\noindent The derivations of Ref.~\cite{Horodecki:2001hcn} provide a proof of this theorem.  
While Ref.~\cite{Horodecki:2001hcn} focuses on the reduction map, all entanglement seen by the reduction map will be seen by the partial transpose map as well~\cite{johnston_ed}. 

A final question remains of whether or not the lowerbound of the maximal $N$ in Theorem~\ref{thm:maxNent} can be further raised, as was performed for the mixture of separable states via derivation of Theorem~\ref{thm:maxNsep}. 
In the following, we provide a constructive framework for entanglement elimination via mixture with $s-1$ entangled states to show that the bound of Theorem~\ref{thm:maxNent} is tight.

For an immediate counterexample showing that $N = s-1$ is capable of eliminating the entanglement for an entangled state of Schmidt rank $s$, consider the generalized Bell basis~\cite{PhysRevLett.70.1895}
\begin{equation}
    |\psi_{nm}\rangle=\sum_j \frac{e^{2\pi ijn/d}}{\sqrt{d}} |j\rangle \otimes |(j+m)  \text{mod}\ d \rangle 
\end{equation}
for $m,n \in {0, \ldots, d-1}$ where $d$ is the Hilbert space dimension on either side of the bipartition.  An equal-probability classical mixture of the set with $m=0$, 
\begin{equation}\label{eq:psin}
    |\psi_n\rangle=\sum_j \frac{e^{2\pi i jn/d}}{\sqrt{d}}|j\rangle \otimes |j\rangle  \ \ \ ,
\end{equation}
results in density matrix
\begin{equation}
\rho=\frac{1}{d}\sum^{d-1}_{n=0}\frac{1}{d}\sum_{j,k}e^{\frac{2\pi i}{d}n(j-k)}|jj\rangle\langle kk| \ \ \ .
\end{equation}
Notice that the density matrix element vanishes when $j$ is not equal to $k$ since $n$ indexes a complete set of roots of unity. However, when $k$ is equal to $j$, the sum results in a factor of $d$. Therefore, the final density matrix will be diagonal and thus PPT. 
The observation that these $d$ entangled pure states, each with Schmidt rank $s = d$, can form a PPT mixed state proves the rigidity of the lowerbound on the maximal $N$ of Theorem~\ref{thm:maxNent} by selecting $|\psi_0\rangle = |\psi_n\rangle$ for any $n$ in the set of Eq.~\eqref{eq:psin}. 

Beyond a single $N = s-1$ counterexample to rule out extension of Theorem~\ref{thm:maxNent}, the entanglement of any state can be eliminated upon mixture with $N = s-1$ specifically selected entangled states.
To see this constructively, let $|\psi_0\rangle$ be a pure state with a Schmidt decomposition of Eq.~\eqref{eq:psi0Schimdt} and density matrix
\begin{equation}
    \rho_0=\sum_i \mu^2_i|ii\rangle\langle ii|+ \sum_{i \neq j} \mu_i\mu_j|ii\rangle\langle jj| \ \ \ .
\end{equation}
The present strategy will be to cancel the off-diagonal terms through mixture, leaving the density matrix diagonal and thus PPT. 
This can be achieved by selecting states 
\begin{equation}
	|\bar{\psi}_i\rangle = \left(-\mu_i |ii\rangle + \sum_{j = i+1}^{s-1} \mu_j |jj\rangle\right)/\gamma_i \ \ \ ,
\end{equation}
where $\gamma_i$ is the normalization factor for $i \in \{ 0, \ldots, s-2\}$, and creating a mixture 
\begin{equation}
\rho = w_0 \rho_0 + \sum_{j=0}^{s-2} w_j |\bar{\psi}_j\rangle \langle \bar{\psi}_j| 
\end{equation}  
with weights $\vec{w} = \left\{ 2^j w_0 \gamma_j^2\right\}_{j = 0, \ldots, s-2}$, where $w_0$ is determined by the unit probability constraint $w_0 + \sum_{j=0}^{s-2} w_j \equiv 1$.
In this procedure, each sequentially added state $|\bar{\psi}_i\rangle$ removes all off-diagonal terms at matrix element locations of $|ii\rangle\langle jj|$ and $|jj\rangle\langle ii|$ for all $j$.
As a constructive strategy of mixing $s-1$  entangled states to produce a PPT mixture, this procedure proves that the boundary in Theorem~\ref{thm:maxNent} is tight for an arbitrary $|\psi_0\rangle$ entangled state. 

It is interesting to momentarily observe that at each step of this procedure, the number of negative eigenvalues in the partial transpose of the mixture will decrease at a rate governed by the maximum Schmidt rank of pure states at diminishing dimensions, i.e., along the tower of $\frac{d(d-1)}{2}$.  
Note that the negative eigenvalues of a mixture can also decrease at an initially faster rate than this by an alternate selection of states. 
For example, the maximally entangled state with Schmidt spectrum $\vec{\mu} = 1/\sqrt{d}$ in even dimensions with $d>2$ mixed in equal probability with a state in which half the amplitudes have been negated results in a density matrix with only $2\binom{d/2}{2} < \frac{(d-1)(d-2)}{2}$ negative eigenvalues.  

\section{Discussion}

Retaining quality quantum correlations permeating quantum devices as they interact with their experimental environments is crucial for quantum information processing.   
In this work, we consider the qudit-qudit negativity rank as a source of insight into the robustness of negativity and inseparability with respect to noise channels at the two extremes of diluting an ensemble with individual or a uniform distribution of quantum states.

In the former case~(\S\ref{sec:sepmix}), the negativity rank of an entangled pure state is shown to be exactly the maximum number of arbitrary separable states that can be included in arbitrary mixture while the ensemble remains NPT.    
This result achieves a tight bound that determines the maximum by proving that the negativity rank of a pure state is invariant upon arbitrary mixture with any separable state. 
Though demonstrated for a narrowly defined quantum channel, the structural incompatibility that produces these invariances of negativity rank or NPT classification may yield further opportunities for embedding quantum information into subspaces that retain entanglement amid logical noise channels.

Note that this result does not imply that there will always exist a mixture of $\left(n_-+1\right)$ separable states resulting in a PPT ensemble. 
While this is tempting to conclude given that the transpose exhibits the maximal entangled subspace dimension~\cite{PhysRevA.87.054301, PhysRevA.87.064302}(\S\ref{sec:NPTsubspaces}), corresponding to its strength in detecting entanglement~\cite{Johnston:2019muf},
the proof of Section~\ref{sec:sepmix} does not rule out the existence of particular NPT states that may have greater resilience towards separable mixture than presently proven for general NPT states. 
Continued study of negativities corresponding to PPT-sensitive PnCP maps~\cite{CHOI197595, Breuer:2006azi, Hall_2006,  Miller:2015hqi, Buckley:2021isn, Masse:2026vps} will further contribute to understanding entanglement classes in terms of their noise robustness.

In the latter case~(\S\ref{sec:negativityDecay}), the negativity rank is found to set the rate of negativity decay under global depolarizing noise.
By identifying a state with the slowest depolarized negativity decay~(\S\ref{sec:APPTboundary}), this rate has delivered an intuitive rederivation~\cite{Arunachalam:2015ywm} of the APPT boundary for ensembles radially connected to pure states on the Bloch surface (a set represented modulo unitary transformations by the isotropic states~\cite{Horodecki:1997av}).  
Because such ensembles have equivalent APPT and ASEP subspaces~\cite{Arunachalam:2015ywm}, this connection provides a succinct calculation---for a set of ensembles to which quantum computations can be isolated via Clifford twirling~\cite{Zyczkowski:2005jvo, Dankert:2009yux}---of the subspace in which unitaries can no longer generate entanglement.
In practice, this role of the negativity rank may inform decisions whether to algorithmically depolarize quantum noise, for example, when simulating spacelike observables in lattice quantum fields where the negativity rank systematically increases as the continuum limit is approached~\cite{Klco:2021cxq,Gao:2024esb}.

The present depolarizing noise analysis also generates an analytic derivation for the maximum purity of a qutrit-qutrit APPT ensemble~(\S\ref{sec:maxPurityAPPT33}).
This result confirms previous numerical calculations~\cite{Dung:2025xjg}, proves the qutrit-qutrit case of a recent conjecture~\cite{Ahiable:2026how}, and is observed to coincide with the maximum purity realized by a qutrit-qutrit ASEP ensemble. Though not conclusive, this observation further supports the possibility that the APPT and ASEP subspaces may be equivalent beyond the qubit-qudit systems currently proven~\cite{Johnston:2013wey, Arunachalam:2015ywm, Abellanet-Vidal:2024jnc, Song:2024svl}.
Combining the upper bound of Theorem~\ref{thm: theoremII.6} with the lower bound of Ref.~\cite{Gurvits:2002yud} provides additional constraints---the purity of a qutrit-qutrit APPT bound entangled ensemble (if it exists) must be between $\frac{1}{8}$ and $\frac{17}{121}$---that can be employed in optimization procedures in search of ensembles that would distinguish the APPT and ASEP subspaces.
In fact, this tight purity upperbound is found to be sufficient to eliminate some recent PPT-sensitive entanglement witnesses from their candidacy of being sensitive to APPT entanglement.
More broadly, these applications illustrate one role of negativity rank analyses for entanglement detection and structural classification in addition to characterizing susceptibility to noise.

Given the multifaceted value of the analytic APPT purity bound, it would be of interest to extend the derivation to higher dimensions.  
Unfortunately, the proof provided in Section~\ref{sec:maxPurityAPPT33} relies on a necessary and sufficient criterion for APPT classification based on linear matrix inequalities~\cite{Hildebrand:2005eni} that experiences rapid growth in complexity when extended to higher dimensions~\cite{OEIS_A237749, Hildebrand:2005eni}.
Despite this challenge, it may be prudent to explore whether components of the proof, e.g., the limited rank available for inequality violation established in Theorem~\ref{thm: negeigL}, persists for the increasing number of corresponding inequalities in qudit-qudit systems with $d >3$.

The entanglement structures exhibited by quantum computations as they evolve are naturally algorithm dependent. 
For quantum field simulations, this includes dependency not only on state preparation and time evolution strategies, but also on fundamental decisions in the representation of a quantum field on quantum hardware, e.g., how the role of entanglement in generating dynamics that satisfy Gauss's law in lattice gauge theories changes when utilizing a basis with explicit or integrated projection quantum numbers~\cite{Ciavarella:2022zhe}.
In order to retain the entanglement necessary for completing their information processing tasks, quantum computations may benefit from designing intermediate entanglement structure to be robust to prevailing noise.
For distillable entanglement in particular, the present analyses indicate that the negativity rank provides guidance for defining specifications of strategically tailored entanglement. 
As our understanding of entanglement properties in the context of noise advances, opportunities to extend the scientific reach of practical quantum devices arise in the form of incorporating entanglement structure into quantum simulation designs.

\begin{acknowledgments}
The authors thank Kenneth R. Brown and Iman Marvian for discussions.  This work was supported by National Science Foundation CAREER award number PHY-2541830.
\end{acknowledgments}

\appendix
\section{Proof of Theorem II.2}
\label{app: Proof of Theorem II.2}
In this section, we provide a proof of Theorem~\ref{thm: negeigL}, that $L_{1,2}$ of Eq.~\eqref{eq:APPTcriteria} have maximally one negative eigenvalue each.  
Given that $\rho$ is PSD and the vector $v = \frac{1}{\sqrt{2}} \begin{pmatrix}
1,-1,0
\end{pmatrix}$ yields $v^T L_{1,2} v>0$ for $\rho$ of any rank (1 to 9), the possibility of $L_{1,2}$ being negative definite or negative semidefinite is immediately excluded.  
Thus, the remaining content of the proof will focus on also excluding the presence of two negative eigenvalues.

Observe that the sparsity of $L_{1,2}$ depends on the rank of the density matrix. For each rank, the number of negative eigenvalues in $L_{1,2}$ can be determined from either exact calculation or evaluation of determinants.
\paragraph{Rank 1}
This corresponds to pure states with an eigenspectrum of $\vec{\lambda} = \left\{1, 0, \ldots, 0\right\}$. From direct calculation, $L_{1,2}$ each have a single negative eigenvalue, $\mu_-=-1$. 
\paragraph{Rank 2}
From direct calculation, the eigenvalues of $L_1$ and $L_2$ are $\mu={\pm \sqrt{\lambda_1^2+\lambda_2^2}, 0}$. 
\paragraph{Rank 3-7}
By direct calculation, the determinants of $L_{1,2}$ are term-by-term negative and non-vanishing, as demonstrated in Table~\ref{tab:L12-common} where $L_1 = L_2$ and Table~\ref{tab:L12-different}.
\begin{table}[h]
\centering
\small
\renewcommand{\arraystretch}{1.15}
\begin{tabular}{cc}
\toprule
\textbf{Rank} & $\det(L_1)=\det(L_2)$ \\
\midrule
3 &
$-2\lambda_1\lambda_2\lambda_3$  \\

4 &
$-2\lambda_1\lambda_2\lambda_3-2\lambda_1^2\lambda_4$  \\

5 &
$-2\lambda_1(\lambda_1 \lambda_4+\lambda_2(\lambda_3-\lambda_5))$  \\
\bottomrule
\end{tabular}
\caption{Determinants of $L_{1,2}$ for ranks 3-5.}
\label{tab:L12-common}
\end{table}

\begin{table}[t]
\centering
\scriptsize
\renewcommand{\arraystretch}{0.6}
\setlength{\tabcolsep}{3pt}

\begin{tabular}{c|c|c}
\toprule
\textbf{Rank} & 6 & 7 \\
\midrule

$\det L_1$
&
$\begin{aligned}[t]
&-2\bigl(\lambda_1^2\lambda_4 \\
&+\lambda_1(\lambda_3-\lambda_5)
   (\lambda_2-\lambda_6)\bigr)
\end{aligned}$
&
$\begin{aligned}[t]
&-2\bigl(\lambda_1^2\lambda_4 \\
&+\lambda_1(\lambda_3-\lambda_5)
   (\lambda_2-\lambda_6)\\
&+(\lambda_2-\lambda_6)^2\lambda_7\bigr)
\end{aligned}$
\\
\midrule

$\det L_2$
&
$\begin{aligned}[t]
&-2\bigl(\lambda_1^2\lambda_4 \\
&+\lambda_1\lambda_2
   (\lambda_3-\lambda_5)\\
&+2\lambda_2^2\lambda_6\bigr)
\end{aligned}$
&
$\begin{aligned}[t]
&-2\bigl(\lambda_1^2\lambda_4 \\
&+\lambda_1(\lambda_3-\lambda_5)
   (\lambda_2-\lambda_7)\\
&+\lambda_6(\lambda_2-\lambda_7)^2\bigr)
\end{aligned}$
\\
\bottomrule
\end{tabular}

\caption{Determinants of $L_1$ and $L_2$ for ranks 6-7.}
\label{tab:L12-different}
\end{table}
\noindent
Since each determinant is negative and $L_{1,2} \nleq 0$, $L_{1,2}$ corresponding to a density matrix with rank between 3-7 must have exactly one negative eigenvalue.
\paragraph{Rank 8}

The determinants of $L_{1,2}$ are negative with the exception of when the spectrum is flat $\lambda_{1, \ldots, 8} = 1/8$ and the determinant vanishes.  In this special case, the eigenvalues of $\vec{\mu} = \{\frac{1}{4},\frac{1}{4},0\}$. Thus, if there exists a negative eigenvalue, there can be maximally one.

\paragraph{Rank 9}
The flips in the sign sequence of leading principal minor determinants~\cite{Horn_Johnson_1985} can be used to rule out the presence of two negative eigenvalues. To begin, the first one-dimensional principal minor is positive as $\lambda_9 > 0$. Because $L_{1,2}$ is observed to necessarily have negative determinant if its two-dimensional leading principal minor has negative determinant, the $\{+,-,+\}$ sign sequence never arises, thus excluding the presence of two negative eigenvalues.

\emph{Rank 9 (Alternative)} Let us focus on $L_1$ to describe an alternate proof for rank 9, whose logic applies also to $L_2$.

Notice that $L_1$ is real, symmetric, and can be separated as $L_1 = A-B$ where $A = \text{diag}\left(L_1\right)$ and $-B$ is the traceless matrix of off-diagonal elements.
Since the set of $\lambda$ eigenvalues are ordered, all elements of $A$ and $B$ are nonnegative.
Because $A$ is nonsingular for rank 9, $D = A^{1/2}$ can be used to decompose 
\begin{equation}
  L_1 = D(\mathbb{I} - C)D \ \ \ ,
\end{equation}
where $B = DCD$. Due to the vanishing diagonal of $B$ and the diagonal structure of $D$, the matrix $C$ will continue to be traceless and symmetric. By Sylvester's law of inertia~\cite{Sylvester01081852}, $(\mathbb{I}-C)$ and $L_1$ will have the same number of negative eigenvalues, so this proof may use $\left(\mathbb{I}-C\right)$ as a proxy for $L_1$.

Let $\omega_1 \geq \omega_2 \geq \omega_3$ be the ordered eigenvalues of $C$.  Since $C$ is populated with entirely nonnegative entries, the Perron-Frobenius theorem~\cite{perron1907grundlagen, frobenius1912matrizen} ensures that the largest (possibly degenerate) eigenvalue is nonnegative, $\omega_1\geq0$.  Because $C$ is also traceless, a vanishing $\omega_1$ would require nonpositive $\omega_{1,2}$ to satisfy $\omega_2+ \omega_3  =0$, and thus $C=0$ and $\vec{\mu} = 2\left\{\lambda_9, \lambda_7, \lambda_4\right\}$ are all positive.  In the alternate case that $\omega_1>0$, the vanishing trace leads to the requirement that $\omega_2 + \omega_3 <0$.    One can also verify that the determinant of $C$ is always nonnegative, leading to $\omega_2\omega_3\geq0$.  Combining these constraints results in $\omega_3<0$ and $\omega_2 \leq 0$.  Thus, $\omega_1$ is the only eigenvalue that can be positive, yielding maximally one negative eigenvalue of $(\mathbb{I}-C)$ as $1-\omega_1$.  Therefore, $L_1$ will also have at most one negative eigenvalue.

Thus we have proven for all possible cases that $L_1$ and $L_2$ can have at most one negative eigenvalue.

\bibliography{biblio}

\end{document}